\documentclass[10pt]{article} 
\usepackage{lmodern}
\usepackage[utf8,utf8x]{inputenc}
\usepackage[a4paper]{geometry}
\usepackage{amsmath,amssymb}
\usepackage{psfig,graphicx}
\usepackage[normal]{caption}
\usepackage{enumitem}
\setlist{itemsep=.15em}
\global\let\Box\square

\newlength{\algowidth}
\def\tabulation{xxx\=xxx\=xxx\=xxx\=xxx\=xxx\=xxx\=xxx\=xxx\= \kill}

\newtheorem{thm}{Theorem}[section] 
 
\newtheorem{lem}{Lemma}[section]

\newtheorem{rem}{Remark}
\newenvironment{proof}{\begin{trivlist} 
        \item[]\hspace{0cm}{\bf Proof.} 
\hspace{0cm} }{\hfill $\Box$ 
                       \end{trivlist}}

\def\bibfmta#1#2#3#4{{#1}, {#2}, {\sl #3}, #4.}
\def\bibfmtb#1#2#3#4{{#1}, {\sl #2}, { #3}, #4.}

\begin{document}

\title{A Distributed Algorithm for Constructing a Minimum Diameter Spanning Tree}
\author{Marc Bui\,$^a\quad$ Franck Butelle\,$^b\,$\thanks{Corresponding author:  
LIPN, CNRS UPRES-A 7030, Universit\'e Paris-Nord, 99, Av. J.-B. Cl\'ement 93430 Villetaneuse, France. E-mail: 
butelle@lipn.univ-paris13.fr}$\quad$ Christian Lavault\,$^c$\\
$^a\,${\small \sl LDCI, Universit\'{e} Paris 8, France}$\quad$  
$^b\,${\small \sl LIPN -- CNRS  7030, Universit\'{e} Paris 13, France}\\
$^c\,${\small \sl LIPN -- CNRS UMR 7030, Universit\'e Paris-Nord}
}
\date{\empty}
\maketitle

\begin{abstract}
We present a new algorithm, which solves the problem of distributively 
finding a minimum diameter spanning tree of any (non-negatively) 
real-weighted graph $G = (V,E,\omega)$. As an intermediate step, we use a new, 
fast, linear-time all-pairs shortest paths distributed algorithm to find 
an absolute center of $G$. The resulting distributed algorithm is asynchronous, 
it works for named asynchronous arbitrary networks and achieves 
$\mathcal{O}(|V|)$ time complexity and $\mathcal{O}\left(|V|\,|E|\right)$ 
message complexity.

\smallskip \noindent 
{\it Keywords}:\ Spanning trees; Minimum diameter spanning trees; Shortest paths; 
Shortest paths trees; All-pairs shortest paths; Absolute centers.
\end{abstract}

\section{Introduction} \label{intro}
Many computer communication networks require nodes to broadcast 
information to other nodes for network control purposes; this is 
done efficiently by sending messages over a spanning tree of the 
network. Now, optimizing the worst-case message propagation over 
a spanning tree is naturally achieved by reducing the diameter to a minimum.

The use of a control structure spanning the entire network is a 
fundamental issue in distributed systems and interconnection networks. 
Given a network, a distributed algorithm is said to be {\em total} 
iff all nodes participate in the computation. Now, all total 
distributed algorithms have a time complexity of $\Omega(D)$, 
where $D$ is the network diameter (either in terms of hops, or 
according to the wider sense given by Christophides in~\cite{Chri75}). 
Therefore, having a spanning tree with minimum diameter of arbitrary
networks makes it possible to design a wide variety of time-efficient 
distributed algorithms. In order to construct such a spanning tree, 
all-pairs shortest paths in the graph are needed first. Several 
distributed algorithms already solve the problem on various assumptions. 
However, our requirements are more general than the usual ones. 
For example, we design a ``process terminating'' algorithm for 
(weighted) networks with no common knowledge shared between the 
processes. (See assumptions in Subsection~\ref{prob} below.)

\subsection{Model, Notations and Definitions}
A distributed system is a standard point-to-point asynchronous 
network consisting of~$n$ communicating processes connected by 
$m$ bidirectional channels. Each process has a local non-shared memory 
and can communicate by sending messages to and receiving messages 
from its neighbours. A single process can transmit to and receive 
from more than one neighbour at a time.

The network topology is described by a finite, weighted, connected 
and undirected graph $G=(V,E,\omega)$, devoid of multiple edges and 
loop-free. $G$ is a structure which consists of a finite set of nodes 
$V$ and a finite set of edges $E$ with real-valued weights; each edge 
$e$ is incident to the elements of an unordered pair of nodes $(u,v)$. 
In the distributed system, $V$ represents the processes, while $E$ 
represents the (weighted) bidirectional communication channels 
operating between neighbouring processes~\cite{Lava95}. 
We assume that, for all $(u,v)\in E$, $\omega(u,v)=\omega(v,u)$ and, 
to shorten the notation, $\omega(u,v)=\omega(e)$ denotes the 
real-valued weight of edge $e=(u,v)$. (Assumptions on real-valued 
weights of edges are specified in the next two Subsections~\ref{prob} 
and~\ref{ass}.) Throughout, we let $|V|=n$, $|E|=m$ and, according 
to the context, we use $G$ to represent the network or the weighted 
graph, indistinctly.

The weight of a path $[u_{0},\ldots,u_{k}]$ of~$G$ 
($u_{i}\in V,\ 0\le i\le k$) is defined as 
$\sum_{0\le i\le k-1}\omega(u_{i},u_{i+1})$. 
For all nodes $u$ and $v$ in $V$, the {\em distance} from $u$ to $v$, 
denoted $d(u,v)=d_G(u,v)=d(v,u)=d_G(v,u)$, is the lowest weight 
of any path length from $u$ to $v$ in $G$. The largest (minimal) 
distance from a node $v$ to all other nodes in~$V$, denoted 
$ecc(v)=ecc_G (v)$, is the {\em eccentricity} of node~$v$: 
$ecc(v)=\max_{u\in V}d(u,v)$~\cite{Chri75}. An {\em absolute center} 
of $G$ is defined as a node (not necessarily unique) achieving 
the smallest eccentricity in $G$. 

$D=D(G)$ denotes the {\em diameter} of $G$, defined as 
$D=\max_{v\in V} ecc(v)$ (see~\cite{Chri75}) and $R=R(G)$ denotes the 
{\em radius} of $G$, defined as $R=\min_{v\in V} ecc(v)$.
Finally, $\Psi(u)=\Psi_G(u)$ represents the {\em shortest paths tree} 
(SPT) of $G$ rooted at node $u$: $(\forall v\in V)\; d_{\Psi(u)}(u,v)=d(u,v)$. $\Psi(u)$ 
is chosen {\em uniquely} among all the shortest paths trees 
of $G$ rooted at node $u$; whenever there is a tie between 
any two length paths $d(u,v)$, it is broken by choosing the 
path with a second node of minimal identity. The set of all 
SPTs of $G$ is denoted $\Psi =\Psi(G)$. When it is clear 
from the context, the name of the graph is omitted.

In the remainder of the paper, we denote problems as ``the (MDST) 
problem'', ``the (MST) problem'', ``the (GMDST) problem'', etc. 
(see definitions in Subsection~\ref{rwr}). Distributed algorithms 
are denoted in italics, e.g. ``algorithm {\em MDST\/}''. Finally, 
``MDST'', ``APSPs'' and ``SPT'' abbreviate ``minimum diameter spanning 
tree'', ``all-pairs shortest paths'' and ``shortest paths tree'', 
respectively.

\subsection{The Problem} \label{prob}
Given a weighted graph $G=(V,E,\omega)$, the {\bf (MDST) problem} 
is to find a spanning tree of~$G$ of minimum diameter $D$ (according 
to the definition of $D$).

Note that the (MDST) problem assumes $G$ to be a non-negatively 
real-weighted graph (i.e., $\forall e\in E\ \; \omega(e)\in \mathbb{R}^+$). 
Indeed, the (MDST) problem is known to be NP-hard if we allow negative 
length cycles in $G$ (cf.~Camerini {\em et al.}~\cite{CaGM80}).

In spite of the fact that the (MDST) problem requires arbitrary 
non-negative real-valued edges weights, our distributed MDST 
algorithm is process terminating (i.e., a proper distributed 
termination is completed~\cite{Lava95}). This is generally not 
the case on the above requirement. When weights are assumed 
to be real-valued, a common (additional) knowledge of a bound 
on the size of the network is usually necessary for APSPs 
algorithms to process terminate (see e.g.~\cite{Awer85,BeGa92,Lamp82}). 
By contrast, no ``structural information'' is assumed in our 
algorithm, neither topological (e.g., size or bound on the size of the 
network), nor a sense of direction, etc. (see Subsection~\ref{formal}).

\subsection{Assumptions} \label{ass}
In addition to the above general hypothesis of the (MDST) problem, 
we need the following assumptions on the network.
\begin{itemize} 
\item Processes are faultless, and the communication channels 
are faithful, lossless and order-preserving (FIFO).

\item All processes have {\em distinct} identities ({\em ID\/}s). 
($G$ is called a ``named network'', by contrast with ``anonymous 
networks''.) We need distinct {\em ID\/}s to compute the APSPs 
routing tables of $G$ at each process of the network. For the sake of 
simplicity, {\em ID\/}s are also assumed to be non-negative distinct 
integers.

Each process must distinguish between its ports, but has no 
{\em a priori} knowledge of its neighbours {\em ID\/}s. Actually, 
any process knows the {\em ID\/} of a sending process after reception 
of its first message. Therefore, we assume w.l.o.g. (and up to $n-1$ 
messages at most) that a process knows the {\em ID\/} of each of its 
neighbours from scratch (see protocol {\em APSP} in Subsection~\ref{mdst}).

\item Of course, each node also knows the weights of its adjacent 
edges. However, edges weights do not satisfy the triangular 
inequality.

\item Let $\mathcal{A}$ be a distributed algorithm defined on $G$. 
A non-empty subset of nodes of $V$, called {\em initiators}, 
simultaneously start algorithm $\mathcal{A}$. In other words, an 
external event (such as a user request, for example), impels the 
initiators to trigger the algorithm. Other (non-initiating) nodes 
``wake up'' upon receipt of a first message.

\item In a reliable asynchronous network, we measure the 
communication complexity of an algorithm $\mathcal{A}$ 
in terms of the maximal number of messages that are received, 
during any execution of $\mathcal{A}$. We also take into account 
the number of bits in the messages (or message size): this yields the 
``bit complexity'' of $\mathcal{A}$. For measuring the time complexity 
of $\mathcal{A}$, we use the definition of standard time complexity 
given in~\cite{Lava95,Lync96}. Standard time complexity is defined 
on {\em ``Asynchronous Bounded Delay networks''} (ABD networks): 
we assume an upper bound transmission delay time of $\tau$ for {\em each} 
message in a channel; $\tau$ is then the ``standard time unit'' in $G$.
\end{itemize}

\subsection{Related Works and Results} \label{rwr}
The small amount of literature related to the (MDST) problem mostly
deals either with graph problems in the Euclidian plane (geometric
minimum diameter spanning tree: the (GMDST) problem), or with the
Steiner spanning tree construction (see~\cite{HLCW91,IhRW91}). The
(MDST) problem is clearly a generalization of the (GMDST) problem. 
The sequential problem has been addressed by some authors (see for
example~\cite{Chri75}).

Surprisingly, despite the importance of having a {\em MDST} in arbitrary
distributed systems, only few papers have addressed the question of
how to design algorithms which construct such spanning trees. Finding
and maintaining a {\em minimum spanning tree} (the (MST) problem) 
has been extensively studied in the literature 
(e.g.~\cite{Awer87,AwCK90,EITT+92,GaHS83}). More recently, the problem 
of maintaining a {\em small} diameter was however solved in~\cite{ItRa94}, 
and the distributed (MDST) problem was addressed in~\cite{BuBu93,BuLB95}.

\subsection{Main contributions of the paper}
Our algorithm {\em APSP} is a generalization of APSP 
algorithms on graphs with unit weights (weights with value 1) 
to the case of non-negatively real-weighted graphs. To our knowledge, 
our MDST finding algorithm is also the first which {\em distributively\/} 
solves the (MDST) problem~\cite{BuBu93}. The algorithm {\em MDST\/} 
works for named arbitrary network topologies with asynchronous 
communications. It achieves an ``efficient''  $\mathcal{O}(n)$ time complexity and 
$\mathcal{O}(nm(\log n + \log W))$ bits communication complexity, where $W$ 
is the largest weight of a channel. (An $\mathcal{O}(n)$ time complexity may be 
considered ``efficient'', though not optimal, since 
the construction of a spanning tree costs at least
$\Omega(D)$ in time).

The paper is organized as follows. In Section~\ref{algo} we present a
high-level description of the protocol {\em APSP}, a formal design 
of the procedure {\em Gamma\_star} and the algorithm {\em MDST}. 
Section~\ref{anal} is devoted to proofs and complexity analysis of the 
algorithm. Finally, concluding remarks are given in Section~\ref{concl}.

\section{The Algorithm} \label{algo}
\subsection{A High-Level Description}
\subsubsection{Main Issues}
First, we recall in Lemma~\ref{lem:abscenter} that the (MDST) 
problem for a weighted graph~$G$ is (polynomially) reducible 
to the absolute center problem for~$G$. Then, we constructively 
find and compute an absolute center of~$G$ by using its APSPs 
routing tables in Lemma~\ref{lem:hackimi}.

In summary, given a positively weighted graph~$G$, the main 
steps of our algorithm for the (MDST) problem are the following:
\begin{enumerate}
\item The computation of APSPs in $G$;
\item The computation of an absolute center of $G$ (procedure 
{\em Gamma\_star}\/$(e)$ in Subsection~\ref{formal});
\item The construction of a {\em MDST}  of $G$, and the transmission of 
the knowledge of that MDST to each node within the network $G$.
\end{enumerate} 

\subsubsection{Construction of a {\em MDST}}\label{mdst}
The definition of the eccentricity is generalized as follows. 
We view an edge $(u,v)$ with weight $\omega$ as a continuous 
interval of length $\omega$, and for any $0<\alpha <\omega$ 
we allow an insertion of a ``dummy node'' $\gamma$ and replace 
the edge $(u,v)$ by a pair of edges: $(u,\gamma)$ with weight 
$\alpha$ and $(\gamma,v)$ with weight $\omega -\alpha$.

According to the definition, the eccentricity $ecc(\gamma)$ 
of a {\em general node} $\gamma$ (i.e., either an actual node 
of $V$, or a dummy node) is clearly given by 
$ecc(\gamma)=\displaystyle{\max_{z\in V}d(\gamma,z)}$.
A node~$\gamma^*$ such that $ecc(\gamma^*)=\min_{\gamma}ecc(\gamma)$ 
is called {\em an absolute center} of the graph. Recall that $\gamma^*$ 
always exists in a connected graph and that it is not unique in general. 
Moreover, an absolute center of $G$ is usually one of the dummy nodes 
(see Fig.~\ref{fig:mdst}).

\begin{figure}[htb] 
\centerline{ \parbox{0.55\textwidth}{%
\includegraphics[height=5cm]{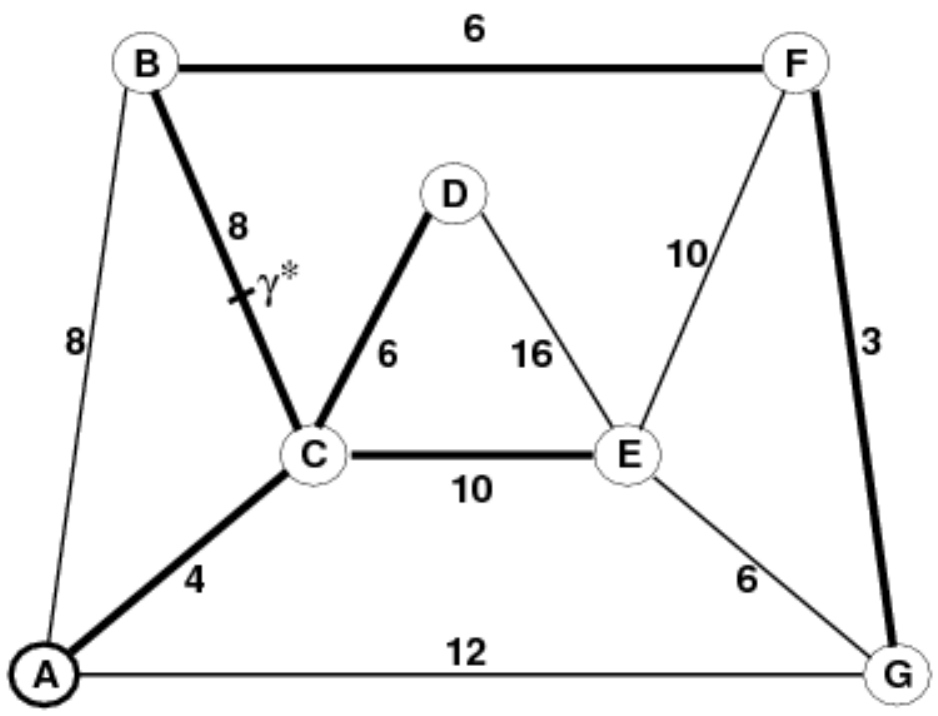}} \ 
\fbox{ \parbox{0.3\textwidth}{%
\rule{.18in}{.02in} \hspace{0.2cm} MDST $T^*$ of $G$ \\ $\gamma^*$: absolute center of $G$}}
}
\caption{Example of a MDST $T^*$ ($D(G)=22$ and $D(T^*)=27$).
$T^*$ is neither a shortest paths tree, nor a minimum spanning tree of~$G$.}
\label{fig:mdst}
\end{figure}
Similarly, the definition of $\Psi(u)$ is also generalized to 
account for the dummy nodes. Finding a {\em MDST} actually reduces 
to searching for an absolute center~$\gamma^*$ of~$G$: the SPT 
rooted at~$\gamma^*$ is a {\em MDST} of~$G$. Such is the purpose of 
the following Lemma~\ref{lem:abscenter}.

\begin{lem} {\em \cite{CaGM80}} \label{lem:abscenter} 
Given a weighted graph G, the (MDST) problem for G is (polynomially) 
reducible to the problem of finding an absolute center of G.
\end{lem}

\subsubsection{Computation of an absolute center of a graph} \label{subsubsect:abscenter}
The idea of computing absolute $p$-centers was first introduced by Hakimi, 
see for example~\cite{HaPS78}. Here we address the computation of an absolute 1-center.
According to the results in~\cite{Chri75}, we need the following lemma
(called Hakimi's method) to find an absolute center of~$G$.

\begin{lem}\label{lem:hackimi}
Let $G=(V,E,\omega)$ be a weighted graph. For each edge $e\in E$, 
let~$\gamma_e$ be the set of all the general nodes of G which achieve a 
minimal eccentricity for~$e$. A node achieving the minimal eccentricity 
among all nodes in $\displaystyle{\bigcup_{e\in E} \gamma_e}$ is an 
absolute center. Finding a minimum absolute center of G is thus achieved in
polynomial time.
\end{lem}
\begin{proof} (The proof is constructive.)

{\em (i)} For each edge $e=(u,v)$, let $\alpha = d(u,\gamma)$.
Since the distance $d(\gamma,z)$ is the length of a path
$[\gamma,u,\ldots ,z]$ or a path $[\gamma,v,\ldots ,z]$,
\begin{equation} \label{eq:sep}
ecc(\gamma) = \max_{z\in V} d(\gamma,z) = %
\max_{z\in V}\,\min \big(\alpha + d(u,z),\omega(u,v) - \alpha + d(v,z)\big).
\end{equation}
If we plot $f_{z}^+(\alpha)=\alpha+d(u,z)$ and
$f_{z}^-(\alpha)=-\alpha +\omega(u,v)+d(v,z)$ in Cartesian coordinates
for fixed $z=z_0$, the real-valued functions $f_{z_0}^+(\alpha)$ and
$f_{z_0}^-(\alpha)$ (separately depending on $\alpha$ in the range
$[0,\omega(e)]$) are represented by two line segments $(S_1)_{z_0}$
and $(S_{-1})_{z_0}$, with slope $+1$ and $-1$, respectively. For a
given $z=z_0$, the smallest of the two terms $f_{z_0}^+(\alpha)$ and
$f_{z_0}^-(\alpha)$ in~(\ref{eq:sep}) define a piecewise linear function 
$f_{z_0}(\alpha)$ made of $(S_1)_{z_0}$ and $(S_{-1})_{z_0}$.

Let $B_e(\alpha)$ be the {\em upper boundary\/} ($\alpha\in
[0,\omega(e)]$) of all the above $f_z(\alpha)$ ($\forall z\in V$). 
$B_e(\alpha)$ is a curve made up of piecewise linear segments, which passes through several local minima
(see Fig.~\ref{fig:bound}). A point $\gamma$ achieving the
smallest minimal value (i.e. the global minimum) of $B_e (\alpha)$
is an absolute center~$\gamma^*_e$ of the edge~$e$.

\medskip {\em (ii)} From the definition of $\gamma^*_e$, $\min_\gamma
ecc(\gamma)=\min_{\gamma^*_e}s(\gamma^*_e)$; and $\gamma^*$ 
achieves the minimal eccentricity. Therefore, an absolute center 
$\gamma^*$ of the graph is found at any point where the minimum of 
all $ecc(\gamma^*_e)$\/s is attained.
\end{proof}

\medskip
\begin{figure}[htb] 
\centerline{ \parbox{0.5\textwidth}{%
\includegraphics[height=5cm]{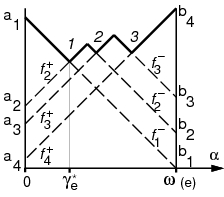}} 
\fbox{ \parbox{0.3\textwidth}{%
\rule{.15in}{.02in} \hspace{0.1cm} $(a_i,b_i)$ pairs of distances\\
$f_i^+(\alpha) = \alpha + a_i$ \\ $f_i^-(\alpha) = \omega(u,v )- \alpha + b_i$}}
}
\caption{Example of an upper boundary~$B_e(\alpha)$}
\label{fig:bound}
\end{figure}
By Lemma~\ref{lem:hackimi}, we may consider this method from an algorithmic viewpoint. 
For each $e=(u,v)\in E$, let
\[
\mathcal{C}_e = \big\{(d_1,d_2) \;|\ \ \forall z\in V\ \ d_1 = d(u,z)\ \ \mbox{and}\ \ d_2 = d(v,z)\big\}.\]
Now, a pair ($d_1$'$,d_2$') is said to {\em dominate} a pair 
$(d_1,d_2)$ iff $d_1\le d_1$' and $d_2\le d_2$'; namely, the 
function $f_{z'}(\alpha)$ defined by ($d_1$'$,d_2$') is over 
$f_z(\alpha)$ defined by $(d_1,d_2)$. Any such pair $(d_1,d_2)$ 
will be ignored when it is dominated by another pair ($d_1$'$,d_2$'). 
The local minima of the upper boundary $B_e (\alpha)$ (numbered 
from 1 to 3 in Figure~\ref{fig:bound}) are located at the intersection 
of the segments $f_{i}^{-}(\alpha)$ and $f_{i+1}^{+}(\alpha)$, 
when all dominated pairs are removed. Sorting the set ${\mathcal C}_e$ 
in descending order, with respect to the first term of each remaining 
pair $(d_1,d_2)$, yields the list $L_e=((a_1,b_1),\ldots,(a_{|L_e|},b_{|L_e|}))$ 
consisting of all such ordered dominating pairs. Hence, the smallest minimum 
of $B_e(\alpha)$ for a given edge $e$ clearly provides an absolute center~$\gamma^*_e$ 
(see the Procedure {\em Gamma\_star}\/$(e)$ in Subsection~\ref{formal}). 
By Lemma~\ref{lem:hackimi}, once all the $\gamma^*_e$\/s are computed, 
we can obtain an absolute center~$\gamma^*$ of the graph~$G$. 
Last, by Lemma~\ref{lem:abscenter}, finding a {\em MDST} of $G$ reduces 
to the problem of computing~$\gamma^*$.

\subsubsection{All-Pairs Shortest Paths Algorithm (protocol {\em APSP})}
In \S\ref{subsubsect:abscenter}, we consider the distances $d(u,z)$ and $d(v,z)$, 
for all $z\in V$ and for each edge $e=(u,v)\in E$. The latter 
distances must be computed by a distributed (process terminating) 
routing algorithm; the protocol {\em APSP} is designed for that 
purpose in Subsection~\ref{formal}.

\subsubsection{Construction and knowledge transmission of a {\em MDST}} \label{subsubsect:constMDST}
At the end of the protocol {\em APSP}, every node knows the node~$u_{min}$ 
with the smallest {\em ID\/} and a shortest path in $G$ leading to $u_{min}$. 
Now, consider the collection of all paths $[u,\ldots,u_{min}]$ (computed by 
{\em APSP}), which start from a node $u\in V$ and end at node $u_{min}$. 
This collection forms a tree rooted at node $u_{min}$ and, since it is an SPT 
of $G$, the information is exchanged ``optimally'' in the SPT 
$\Psi(u_{min})$\footnote{In $\Psi(u_{min})$, the information is transmitted 
{\em ``optimally''} in terms of time and messages, in the sense that 
each edge weight may be regarded as the message transmission delay of 
a channel.}. Hence, the number of messages needed to search an extremum 
in the tree $\Psi(u_{min})$ is at most $\mathcal{O}(n)$ (with message size 
$\mathcal{O}\left(\log n + \log W\right)$).

When the computation of an absolute center $\gamma^*$ of $G$ 
is completed, the endpoint of $\gamma^*$'s edge having the smallest 
{\em ID\/} sends a message to $u_{min}$ carrying the {\em ID\/} of~$\gamma^*$. 
Upon receipt of the message, $u_{min}$ forwards the information all over 
$\Psi(u_{min})$ (adding the same cost in time and messages). 
Therefore each node of $G$ knows a route to~$\gamma^*$, 
and the MDST is built as a common knowledge for all nodes.

\subsection{The Design of the Algorithm {\em MDST}} \label{formal}

\subsubsection{Main Procedure}
The distributed algorithm {\em MDST} finds a MDST of an input 
weighted graph $G = (V,E,\omega)$ by computing an absolute center of~$G$. 

The algorithm is described from a node point of view. The algorithm assumes that each node~$u$
performs the following steps.
\begin{description}
\item[Step 1.]\ Node $u$ participates in the computation of the APSP. 
This computation gives the diameter and the radius of the graph $G$. 
Moreover it also gives $u_{min}$, the minimum node identity in the graph. 
(See \S\ref{subsubsect:constMDST}.) 

\item[Steps 2 \& 3.]\ An adjacent edge selection procedure is implemented by discarding heavy edges. 
The computation of the local minimum is accelerated with the help of an upper bound test. 
Note that the variable $\varphi$, used in the test, is a data structure with four fields: 
the best distance $\alpha$ from the first edge end, the upper bound value associed to $\alpha$, 
the identities of the first and second edge ends. (Edge ends are ordered by increasing identities.)

\item[Steps 4, 5 \& 6.]\ Node $u$ participates in finding the minimum of all values~$\varphi$. 

\item[Step 7]\ The best $\varphi$ is finally computed at the root of the tree $\Psi(u_{min})$ 
and next, it is broadcast to all nodes through $\Psi\left(u_{min}\right)$.
\end{description}
For the sake of clarity, we use abstract record data types (with dot notation).

\bigskip \noindent \hrulefill \; {\bf Algorithm {\em MDST} (for node $u$)}\; \hrulefill

\vspace{-1mm}
\begin{tabbing}xxx\=xxx\=xxx\=xxx\=xxx\=xxx\=\kill\>
{\bf Type} elt: {\bf record}\+\\
\>\>\>  $alpha\_best$, $upbound$: EdgeWeight;\+\\
\>\> $id_1$, $id_2$: NodeIdentity;\\
\>\>	{\bf end};\-\\[.3\baselineskip] 
{\bf Var}\ $\varphi$, $\varphi_u^*$: elt;\+\\ 
	$Diam$, $Radius$, $\alpha$, $localmin$: EdgeWeight;\\
     $u_{min}$: NodeIdentity;\\
	$d_u$: array of EdgeWeight; \`{\small \em (*~after step 1, $d_{u}[v]=d(u,v)$~*)}
\end{tabbing}
\vspace{-1mm}
\begin{enumerate}
\item[(1)]\ {\bf For all} $v\in V$ Compute $d_u[v]$, $Diam$, $Radius$ and $u_{min}$; 
\hfill {\small \em (*~from protocol APSP~*)}
\item[(2)]\ $\varphi.upbound \gets Radius$; 
\item[(3)]\ {\bf While} $\varphi.upbound>Diam/2$ {\bf do for} each edg $e=(u,v)$ s.t. $v > u$
\vspace{-1mm}
 \begin{enumerate}
 \item[(a)]\ $(\alpha,localmin) \gets$ {\em Gamma\_star}\/$(e)$; 
 \item[(b)]\ {\bf If} $localmin<\varphi.upbound$ {\bf then} $\varphi\gets (\alpha,localmin,u,v)$; 
\end{enumerate}
\item[(4)]\ $\varphi_u^* \gets \varphi$;
\item[(5)]\ {\bf Wait for reception of} $\varphi$ 
from each child of $u$ in $\Psi(u_{min})$ {\bf and do} 
\vspace{-1mm}
\begin{description}
\item[\ if]\ $\varphi_u^*.upbound<\varphi.upbound$\ {\bf then}\ $\varphi_u^* \gets \varphi$;
\end{description}
\item[(6)]\ {\bf Send} $\varphi_u^*$ to parent in $\Psi(u_{min})$; 
\item[(7)]\ {\bf If}\ $u = u_{min}$\ {\bf then Send} $\varphi_u^*$ to all its children
\vspace{-1mm}
  \begin{description}
  \item[else]\ {\bf Wait for reception of} $\varphi^*$ from its parent 
  {\bf then Send} $\varphi^*$ to all its children
  \end{description}
\end{enumerate}
\vspace{-1mm}

\noindent \hrulefill\-\hrulefill

Now we describe the basic procedures used in the algorithm: first the
protocol {\em APSP} and next the procedure {\em Gamma\_star}\/$(e)$.

\subsubsection{The APSP algorithm}
We need an algorithm that computes the all-pairs shortest paths 
in $G$ and does process terminate without any structural information 
(e.g., the knowledge an upper bound on~$n$). Our algorithm is 
based on the Netchange algorithm (see the proof in~\cite{Lamp82}), 
the Bellman-Ford algorithm (see~\cite{BeGa92}) and the 
$\alpha$-synchroniser described in~\cite{Awer85}. The three latter 
algorithms process terminate {\em iff\/} an upper bound on~$n$ 
is known. Otherwise, if the processes have no structural information, 
the above algorithms only ``message terminate'' (see~\cite{Awer85,Lava95}). 
However, proper distributed termination may be achieved without 
additional knowledge by using the same technique as designed 
in~\cite{ChVe90}. We now shortly describe the algorithm (from the 
viewpoint of node $u$, whose {\em ID\/} is $id_u$).

The protocol {\em APSP\/} is organized in phases after the first 
initialization step. This step starts initializing sets and variables 
($id_u$ is the selected {\em ID\/}): the distance to $id_u$ is set to 0, 
while all others distances are set to $\infty$ and the set {\em Updatenodes\/} 
is initialized to $\emptyset$. Next, every phase of the algorithm consists 
of three steps.
\begin{description}
\item[Step 1.]\ Send to all neighbours the {\em ID\/} of the selected node 
and its distance to node $u$.
\item[Step 2.]\ Wait for reception of the same number of messages sent 
in step 1 minus the number of inactive neighbours (see next paragraph). 
Upon receipt of a message, update distance tables. If the estimate of the 
distance to a node changes, add this node to the set {\em Updatenodes\/}. 
If an $\langle${\bf Inactive}$\rangle$ message is received from 
a neighbour, mark it {\em inactive\/}. When the awaited number of 
messages is received, start step~3.
\item[Step 3.]\ Choose an active node from the set {\em Updatenodes\/}
with the smallest distance to $u$ and go to step 1. If no such node exists 
then send an $\langle${\bf Inactive}$\rangle$ messsage to each active 
neighbour; node $u$ becomes an inactive node.
\end{description}

We need the following  rules to make the algorithm process terminate.
\begin{enumerate}
\item[(1)]\ An inactive node forwards updating messages (if necessary) to its 
inactive neighbours.
\item[(2)]\ Only one $\langle${\bf Inactive}$\rangle$ message is sent from 
node $u$ to a neighbour $v$ and this message is the last message (of 
protocol {\em APSP}) from $u$ to $v$.
\item[(3)]\ (from the previous rule) A node terminates only when two 
$\langle${\bf Inactive}$\rangle$ messages are received in each 
of its adjacent edges (one from each direction).
\end{enumerate}
\vspace{-1mm}
Thus, we designed a new distributed APSP protocol having a good message complexity: $2mn$.

\subsubsection{Procedure Gamma\_star}
Assume the list $L_e$ (defined in \S\ref{subsubsect:abscenter}) to be already 
constructed (e.g. with a heap) when the routing tables are computed. 
For any fixed edge $e\in E$, the next procedure returns a value $\gamma^*_e$.

\bigskip \noindent \hrulefill\ \;{\bf Procedure {\em Gamma\_star}\/$(e)$}\; \hrulefill

\begin{tabbing}\tabulation\>
{\bf Var} $min, \alpha$: real; \hspace{0.5cm} {\bf Init} $min\gets a_1$; $\alpha\gets 0$;\+\\ 
{\bf For}\ $i\gets 1$\ to\ $|L_e|-1$\ {\bf do}\ \hspace{1cm} {\small \sl (*~Compute the intersection $(x,y)$ 
of segments $f_{i}^{-}$ and $f_{i+1}^{+}$~*)} \\[.3\baselineskip] 
$x\gets \frac{1}{2}\, \big(\omega(e) - a_{i+1} + b_i\big)$; \\[.3\baselineskip]
$y\gets \frac{1}{2}\, \big(\omega(e) + a_{i+1} + b_i\big)$; \\[.3\baselineskip]
{\bf if}\ $y < min$\ {\bf then}\ $min \gets y$; $\alpha \gets x$ ;\-\\[.3\baselineskip]
{\bf Return}($\alpha$,$min$)
\end{tabbing}
\vspace{-1.5mm}

\noindent \hrulefill\-\hrulefill

\begin{rem} {\rm Recall that for each edge $e=(u,v)$ of $G$ with weight $\omega(e)$ and for any given $z\in V$, 
$d_1$ and $d_2$ are the distances $d_1=d(u,z)$ and $d_2=d(v,z)$. 
Moreover, all pairs $(a_i,b_i)$ ($1\le i\le |L_e|$) are those ordered pairs $(d_1,d_2)$ 
of the list $L_e$ which are dominating pairs (see the proof of Lemma~\ref{lem:abscenter}).}
\end{rem}

\section{Analysis} \label{anal}
For the purpose of the complexity analysis, let $W\in \mathbb{R}^+$ 
be the largest weight of all edges in $E$: the number of bits in 
$W$ is $\lceil \log_2 W\rceil$. Therefore, the weight of an edge 
requires $\mathcal{O}(\log W)$ bits and the weight of any path (with no cycle) 
uses $\mathcal{O}(\log(nW))$ bits.

The following Lemma~\ref{lem:apsp} gives the complexity of the
protocol {\em APSP}. Next, the Theorem~\ref{theo:mdst} derives the 
time and the communication complexity of the algorithm {\em MDST\/} 
from Lemma~\ref{lem:apsp}.

\begin{lem}\label{lem:apsp}
The All-Pairs Shortest Paths protocol APSP process terminates. It 
runs in $\mathcal{O}(n)$ time and uses $\mathcal{O}(nm)$ messages to compute the
routing tables at each node of G. Its message size is at most
$\mathcal{O}(\log n+\log (nW))$.
\end{lem}
 
\begin{proof}
The protocol {\em APSP} is almost identical to the well-known 
distributed Bellman-Ford shortest-paths algorithm (except for the notion 
of active/inactive nodes). The following definitions are taken from~\cite{BeGa92}.

Let $S\subseteq V$. A path $[u_0,\ldots,u_k]$ is called  an $S$-path 
if for all $i$ $(0\le i\le k)$, $u_i\in S$. The $S$-distance from $u$ to $v$, 
denoted $d^S(u,v)$, is the smallest weight of any $S$-path that joins $u$ 
to $v$. When $S=V$, we write $d(u,v)=d^V(u,v)$.
As a consequence, for all $z\in V$,
\begin{enumerate}
\item[{\em (i)}]\ If $S' = S\cup \{z\}$, then for all $u,\, v\in S$,
\begin{equation} \label{eq:min}
d^{S'}(u,v)\; \stackrel{def}{=}\; \min \left( d^S(u,v),d^{S'}(u,z)+d^{S'}(z,v) \right). 
\end{equation} 
\item[{\em (ii)}]\ Let {\em Neigh}$_u$ be the set of neighbours of a node $u\in V$. For any $v\in V$,
\end{enumerate}
\begin{equation} \label{eq:dist}
d(u,v)\; \stackrel{def}{=}\; \left\{
\begin{array}{ll}
0 & \text{if}\ \; u = v \\
\min_{z\in Neigh_u} \big(\omega(u,z) + d(z,v)\big) & \text{otherwise}.
\end{array}\right.
\end{equation}
Since the algorithm is built from the definitions~(\ref{eq:min}) 
and~(\ref{eq:dist}), it does converge to the shortest paths 
(see~\cite{BeGa92,Lamp82}). Also, since the communication channels 
are assumed to be FIFO (see~\cite{ChVe90} and Subsection~\ref{ass}), 
the algorithm process terminates. The above rules ensure that no message 
in the protocol {\em APSP\/} is sent to a terminating node.

Our protocol is based on algorithms which are known to converge in 
$n$ phases (see~\cite{BeGa92,Lamp82}). For an active node, a phase takes 
at most two time units in an ABD network (see Subsection~\ref{ass}): 
sending a message to each neighbour and next receiving a message only 
from all active neighbours). To make the protocol {\em APSP\/} process 
terminate we need an $\langle${\bf Inactive}$\rangle$ message: in the 
worst case (for example when $G$ is a line) exchanging $\langle${\bf Inactive}$\rangle$ 
messages between nodes takes $\mathcal{O}(n)$ time units.

The identity of each node is sent from each active node along each of its 
adjacent edges. The number of messages sent from every node $v$ is thus 
$\mathcal{O}(n\delta(v))$, where $\delta(v)$ is the degree of $v$. Inactive nodes 
simply forward update messages to their inactive neighbours and they do not increase 
the message complexity. Therefore, the message complexity of protocol {\em APSP\/} 
is proportional to $\displaystyle 2nm = \sum_v n\delta(v)$~\cite{Lamp82}. 

From the rules of the protocol (in Subsection~\ref{formal}), adding 
all $\langle${\bf Inactive}$\rangle$ messages makes exactly $2m$. 
Finally, the message complexity of protocol {\em APSP\/} is $\mathcal{O}(nm)$. 
Note that each message carries the {\em ID\/} of  the sending node, 
the {\em ID\/} of the selected node and the distance between both nodes.
\end{proof}

\begin{thm}\label{theo:mdst}
The algorithm MDST solves the (MDST) problem for any distributed 
positively weighted network G in $\mathcal{O}(n)$ time. The communication 
complexity of MDST is $\mathcal{O}\left( nm\left(\log n + \log(nW)\right) \right)$ bits, 
and its space complexity is at most $\mathcal{O}\left( n\left(\log n + \log(nW)\right) \right)$ 
bits (at each node). The number of bits used for the ID of a node is $\mathcal{O}(\log n)$, 
and the weight of a path ending at that node is $\mathcal{O}\left(\log nW\right)$.
\end{thm}

\begin{proof} The proof derives readily from the previous lemma and Subsection~\ref{subsubsect:constMDST} \end{proof}

\section{Concluding Remarks} \label{concl}
Given a positively weighted graph $G$, our algorithm {\em MDST\/} 
constructs a {\em MDST} of $G$ and distributively forwards the control 
structure over the named network $G$. This new algorithm is simple 
and natural. It is also time and message efficient: complexity measures 
are $\mathcal{O}(n)$ and $\mathcal{O}(nm)$, respectively, which, in some sense, is ``almost'' 
the best achievable (though not optimal) in a distributed setting.

By contrast, the space complexity seems to be far from satisfactory.
This is a drawback to the very general assumptions used in the algorithm, 
especially the assumptions on universal (APSPs) routings in arbitrary 
network topologies. For example, algorithm {\em MDST\/} needs a grand 
total of $\mathcal{O}\left(n^2 \left(\log n +\log(nW)\right)\right)$ bits to store 
all routing tables in the entire network. Now, it was recently shown that reasonable APSP 
routing schemes require at least $\Omega\left(n^2\right)$ bits~\cite{FrGu95}. 
This is only a logarithmic factor away from the space complexity of  
algorithm {\em MDST}.

\end{document}